\documentclass[conference,letterpaper]{IEEEtran}

\usepackage[utf8]{inputenc} 
\usepackage[T1]{fontenc}
\usepackage{url}
\usepackage{ifthen}
\usepackage{cite}
\usepackage[cmex10]{amsmath} 

\usepackage{balance}
\usepackage{amssymb,amsfonts}
\usepackage{textcomp}
\usepackage{xcolor}
\usepackage{amssymb}
\usepackage{comment}
\usepackage{mathtools} 
\usepackage{amsthm}
\usepackage{amsfonts} 
\usepackage{graphicx}
\usepackage{float}
\usepackage{caption}
\usepackage{subcaption}
\usepackage{algorithm}
\usepackage[noend]{algpseudocode}
\newcommand\numberthis{\addtocounter{equation}{1}\tag{\theequation}}
 \usepackage{bm}
 \usepackage{multirow}

\newtheorem{theorem}{Theorem}
\newtheorem{definition}{Definition}

\newtheorem{remark}{Remark}

\newtheorem{corollary}{Corollary}
\newtheorem{proposition}{Proposition}

\newcommand\independent{\protect\mathpalette{\protect\independenT}{\perp}}
\def\independenT#1#2{\mathrel{\rlap{$#1#2$}\mkern2mu{#1#2}}}

\usepackage{xspace}
\usepackage{bbm}
\usepackage{mathrsfs}
\usepackage{amssymb}

\newcommand{\Ec}{\mathcal{E}}

\newcommand{\Pc}{\mathcal{P}}

\newcommand{\Sc}{\mathcal{S}}
\newcommand{\Tc}{\mathcal{T}}

\newcommand{\Xc}{\mathcal{X}}
\newcommand{\Yc}{\mathcal{Y}}

\DeclareMathOperator\E{\sf E}

\newcommand{\Norm}{\mathcal{N}}

\newcommand{\U}{\mathrm{Unif}}

\def\textiid{i.i.d.\@\xspace}
\newcommand\iid{\ifmmode\text{ i.i.d. } \else \textiid \fi}

\begin{document}
\title{Mining Invariance from Nonlinear Multi-Environment Data: Binary Classification} 


\author{\IEEEauthorblockN{Austin Goddard, Kang Du, Yu Xiang}
\IEEEauthorblockA{\textit{Department of Electrical and Computer Engineering} \\
\textit{University of Utah}\\
\{austin.goddard,\,kang.du,\,yu.xiang\}@utah.edu}
}

\allowdisplaybreaks

\maketitle


\begin{abstract}
 Making predictions in an unseen environment given data from multiple training environments is a challenging task. We approach this problem from an invariance perspective, focusing on binary classification to shed light on general nonlinear data generation mechanisms. We identify a unique form of invariance that exists solely in a binary setting that allows us to train models invariant over environments. We provide sufficient conditions for such invariance and show it is robust even when environmental conditions vary greatly. Our formulation admits a causal interpretation, allowing us to compare it with various frameworks. Finally, we propose a heuristic prediction method and conduct experiments using real and synthetic datasets. 
\end{abstract}

\section{Introduction}

It is common practice to collect observations of a set of features $X=(X_1,\dots, X_m)$ and response $Y$ from different environments to train a model. The prediction of the response in an unseen environment is often referred to as multi-environment domain adaptation, with practical applications in various fields (e.g., genetics~\cite{meinshausen2016methods} and healthcare~\cite{goddard2022invariance}). A common assumption in such problems is the principle of invariance, modularity, or autonomy~\cite{haavelmo1944probability, aldrich1989autonomy,hoover1990logic,scholkopf2012causal,dawid2010identifying,pearl2009causality}. This invariance assumption states that the conditional distribution of $Y$ given $X$ is invariant with respect to different environment.

The invariant causal prediction (\textsf{ICP}) framework~\cite{peters2016causal}, along with its various extensions~\cite{heinze2018invariant,pfister2019invariant}, employ the invariance principle to identify invariant predictors across environments. Following this framework, various domain adaptation approaches have been developed~\cite{rojas2018invariant,rothenhausler2021anchor,pfister2021stabilizing}. Specifically, the stabilized regression (\textsf{SR})~\cite{pfister2021stabilizing} approach relies on a weaker form of invariance dependent on expectation as opposed to probability. The common assumption for the approaches mentioned is that the assignment of $Y$ does not change over environments. In a causal sense, from which much of the literature in this area stems, this is referred to as an \emph{intervention} on $Y$~\cite{pearl2009causality}. When $Y$ is intervened, the invariance principle, as well as the frameworks mentioned above,  fail. In a series of recent works\cite{du2023learning,du2023generalized}, an alternative approach called the invariant matching property (\textsf{IMP}) has been developed to detect \emph{linear} invariant models in a \emph{regression} setting even when the assignment of $Y$ is altered over environment. 


In this work, we extend general principles developed in~\cite{du2023learning,du2023generalized} to the binary classification setting as an attempt to generalize to nonlinear settings. The proposed approach works even when data-generating models change over environments (e.g., $Y$ can be generated using a probit model for one environment and a logistic model in another). Additionally, the approach is not constrained by the data type, meaning it can be useful on continuous, discrete, or categorical variables. 

\section{Problem Formulation}\label{sec::formulation}
For different environmental conditions indexed by the set $\Ec$, we have a random vector $X=(X_1,\dots, X_m)$ and a binary random variable $Y$ whose elements form a joint distribution $\Pc_e := \Pc_e^{X,Y}$ dependent on $e\in\Ec$. Denote $X$ and $Y$ as $X^e$ and $Y^e$ for a specific $e\in\Ec$, respectively. The supports of $X$ and $Y$ are $\Xc=\mathbb{R}^m$ and $\Yc=\{0,1\}$, respectively. Let $X_{S}$ be a random vector containing the elements in $X$ indexed by the set $S\subseteq\{1,\dots,m\}$, and let $\Xc_S$ be its support. To simplify notation, let $X_0^e:=Y^e$. For each $e\in\Ec$, we keep the distribution $\Pc_e$ general, with the exception that there exists an $X^e_i$ generated according to the form 
\begin{equation}\label{equ::Xi_def}
   X^e_i = g(X^e_{S_i}) + \epsilon^e, \text{ for some }i\in\{1,\dots,m\}, 
\end{equation}
where $X^e_{S_i}$, for $S_i\subseteq\{0,\dots,m\}\backslash i$, represents the variables that directly effect $X^e_i$,  and $\epsilon^e$ is an independent, zero mean, noise variable. We assume the output of the function $g$ is not constant with regards to any of its inputs. 

Additionally, while the function $g$ does not change over environment, the distribution of $\epsilon^e$ can change arbitrarily as long as the mean of the distribution remains zero. Aside from a binary $Y$ and the form of $X_i^e$ in~\eqref{equ::Xi_def}, we make no assumptions on the distribution or functional form of any variable. As such, this formulation applies to any set of features, be it continuous, discrete, or a mixture of the two. 

We assume only a subset of all environments are observed and denote this set by $\Ec_{\text{obs}}\subseteq\Ec$. Where $\Ec_{\text{obs}} = \Ec_{\text{train}}\cup\{e^{\text{test}}\}$, and $Y^{\text{test}}:=Y^{e^{\text{test}}}$, our goal is to make predictions on $Y^{\text{test}}$, given a set of training environments $\Ec_{\text{train}}$.
 As such, we aim to find a function $\phi_e:\Xc\xrightarrow{}\mathcal{W}$ such that, the probability of $Y$ given $\phi_e(X)$ does not vary over any environment. Specifically, for all $w\in\mathcal{W}$ and $e,h\in\Ec_{\text{obs}}$,
\begin{equation}\label{equ::invar}
    \Pc_e(Y|\phi_e(X)=w) = \Pc_h(Y|\phi_h(X)=w).
\end{equation}
As $Y$ is binary, it is equivalent to write~\eqref{equ::invar} in the form: $\E_{\Pc_e}[Y|\phi_e(X)=w] = \E_{\Pc_h}[Y|\phi_h(X)=w]$,  for all  $w\in\mathcal{W}$ and $e,h\in\Ec_{\text{obs}}$. It is well-known that~\eqref{equ::invar} is satisfied if $\phi_e(X) = X_{S_Y}$ and for $S_Y\subseteq\{1,\dots,m\}$,
\begin{equation}\label{equ::parentY}
    Y^e = f(X_{S_Y}^e) + \epsilon_Y,
\end{equation} 
where $\epsilon_Y$ is an independent noise that does not vary over environment~\cite{peters2016causal}. However, we are interested in a more general setting where \emph{the function $f$ and distribution of the noise can vary over environment}. From a causal perspective, this would indicate that $Y^e$ had been \emph{intervened} (see Section~\ref{sec::causal}). In such a setting, $\phi_e(X) = X_{S_Y}$ is no longer useful and other approaches must be considered. We now consider one such alternative, starting with a motivating example. 

\section{Motivating Example}

Consider the following setting with $X^e = (X^e_1,X^e_2,X^e_3)$.  Let $X^e_1$ and $X^e_2$ be independent and follow $X^e_1 \sim \Norm(\mu_{1}^e,\sigma^2_{1})$ and $X^e_2 \sim \Norm(\mu_{2}^e,\sigma^2_{2})$. The variable $Y^e$ is generated such that $Y^e|X^e_1,X^e_2$ forms a probit model. Specifically, 
\begin{equation*}
    Y^e  =\begin{cases}
	   1, & \text{if $\beta_1^eX^e_1 + \beta_2 X^e_2 + \epsilon_Y > 0 $},\\
       0, & \text{otherwise}.
	\end{cases}
\end{equation*}
Following a similar form as~\eqref{equ::Xi_def}, $X^e_3$ is linear given $Y^e$ so that
\begin{equation*}
    X^e_3 = \begin{cases}
	   \gamma_1X^e_1  + \epsilon_3, & \text{if $Y^e = 1$},\\
       \gamma_0X^e_1 + \epsilon_3, & \text{if $Y^e = 0$}.
	\end{cases}
\end{equation*}
The noise variables $\epsilon_Y$ and $\epsilon_3$ are \iid $\Norm(0, \sigma^2)$. Suppose we wish to predict $Y^e$ given only $X^e_1$. 
Predicting $Y^e$ for a particular $e\in\Ec$ is difficult as $\beta_1^e$ and $\mu^e_2$ depend on $e$. i.e., 
\begin{equation} \label{equ::pygxx}
    \E_{\Pc_e}[Y|X_1=x_1] = \Phi\left(\frac{\beta_1^ex_1 + \beta_2\mu^e_2}{\sqrt{(\beta_2\sigma_2)^2 + \sigma^2}} \right),
\end{equation}
where $\Phi$ is the cumulative distribution function of a standard normal random variable. As~\eqref{equ::pygxx} varies over environment, it is not practical to use $\E_{\Pc_e}[Y|X_1]$ to estimate $Y^e$ on different environments. Even while conditioning on both $X^e_1$ and $X^e_2$ (the variables that directly affect $Y^e$), the variance (w.r.t. environment) still remains through $\beta_1^e$.

We can, however, decompose~\eqref{equ::pygxx} into various variant and invariant components such that $\E_{\Pc_e}[Y|X_1=x_1]$ becomes the following (see the proof of Proposition~\ref{prop::suff} for a general case),
\begin{align*}
     \frac{\E_{\Pc_e}[X_3|X_1=x_1] - \E_{\Pc_e}[X_3|X_1=x_1,Y=0]}{\E_{\Pc_e}[X_3|X_1=x_1,Y=1] - \E_{\Pc_e}[X_3|X_1=x_1,Y=0]}, \numberthis\label{equ::examp_matching}
\end{align*}
where $\E_{\Pc_e}[X_3|X_1=x_1]$ is 
\begin{align*}
     \Phi\left(\frac{\beta_1^ex_1 + \beta_2\mu^e_2}{\sqrt{(\beta_2\sigma_2)^2 + \sigma^2}} \right) (\gamma_1 - \gamma_0)x_1 + \gamma_0x_1, \numberthis
\end{align*}
and $\E_{\Pc_e}[X_3|X_1=x_1,Y=y]$ is $\gamma_1x_1$ if $y=1$ and $\gamma_0x_1$ if $y=0$. We note that the variance (w.r.t environment) contributed by $\beta_1^e$ and $\mu_2^e$ is completely accounted for in the term $\E_{\Pc_e}[X_3|X_1]$ and that $\E_{\Pc_e}[X_3|X_1,Y]$ is invariant over environment. Thus,~\eqref{equ::invar} holds for the function $\phi_e(X) = (X_1,\E_{\Pc_e}[X_3|X_1])$. In addition to this, we also note that conditioning on both $X_1$ and $X_2$ leads to a similar invariance; we only condition on $X_1$ in this example for simplicity. 

This invariance does not hold if we replace $X^e_3$ with any other variable. For example, suppose we were to estimate $Y^e$, replacing $X^e_3$ with $X^e_2$. We can still decompose~\eqref{equ::pygxx} similarly to~\eqref{equ::examp_matching} by replacing $X^e_3$ with $X^e_2$. As $\E_{\Pc_e}[X_2|X_1] = \mu^e_2$ does not contain $\beta_1^e$, the portion of $\E_{\Pc_e}[Y|X_1]$ that contains $\beta_1^e$ must reside in $\E_{\Pc_e}[X_2|X_1,Y]$. i.e., $\E_{\Pc_e}[X_2|X_1,Y]$ is not invariant over environments as is $\E_{\Pc_e}[X_3|X_1,Y]$. Thus, the function $\phi_e(X) = (X_1,\E_{\Pc_e}[X_2|X_1])$ will no longer  satisfy~\eqref{equ::invar}.

To further illustrate the difference in selecting $X_3^e$ over $X_2^e$, suppose we wish to estimate on a new environment $e^{\text{test}}$. While we have access to $X^{\text{test}}$, we can easily construct $\E_{\Pc_{\text{test}}}[X_i|X_1]$ for either $i\in\{2,3\}$. We cannot, however, use $Y^{\text{test}}$ to construct our estimate, and $\E_{\Pc_{\text{test}}}[X_i|X_1,Y]$ must be obtained by leveraging invariances over environment. Thus, for either $i\in\{2,3\}$, we construct the estimate 
\begin{equation}
    \hat{Y}^{\text{test}}_i =: \frac{\E_{\Pc_{\text{test}}}[X_i|X_1]-\E_{\Pc_e}[X_i|X_1,Y=0]}{\E_{\Pc_e}[X_i|X_1,Y=1]-\E_{\Pc_e}[X_i|X_1,Y=0]},
\end{equation}
where $e\in\Ec_{\text{train}}$. As $\E_{\Pc_e}[X_3|X_1,Y]$ is invariant and $\E_{\Pc_e}[X_2|X_1,Y]$ is not invariant as discussed above, $\hat{Y}^{\text{test}}_3$ will provide a good estimate of $Y^{\text{test}}$, while $\hat{Y}^{\text{test}}_2$ will not. 

In Fig.~\ref{fig:example} we compare  $\hat{Y}^{\text{test}}_3$ and  $\hat{Y}^{\text{test}}_2$ by simulating $(x^{\text{test}},y^{\text{test}})$ pairs for a set of specific parameters. The estimate $\hat{Y}^{\text{test}}_2$ does not fit the data as many $x_1^{\text{test}}$ corresponding to $y^{\text{test}}=0$ will be incorrectly classified to one. However, this is not the case when $\hat{Y}^{\text{test}}_3$ is used, and the fit is greatly improved (Fig.~\ref{fig:example}). The poor fit on $\hat{Y}^{\text{test}}_2$ is a result of $\E_{\Pc_e}[X_2|X_1,Y]$ varying across environments. 

\begin{center}
\begin{figure}[h]
  \begin{subfigure}[b]{0.45\columnwidth}   
\vspace{-2em}\includegraphics[width=\linewidth]{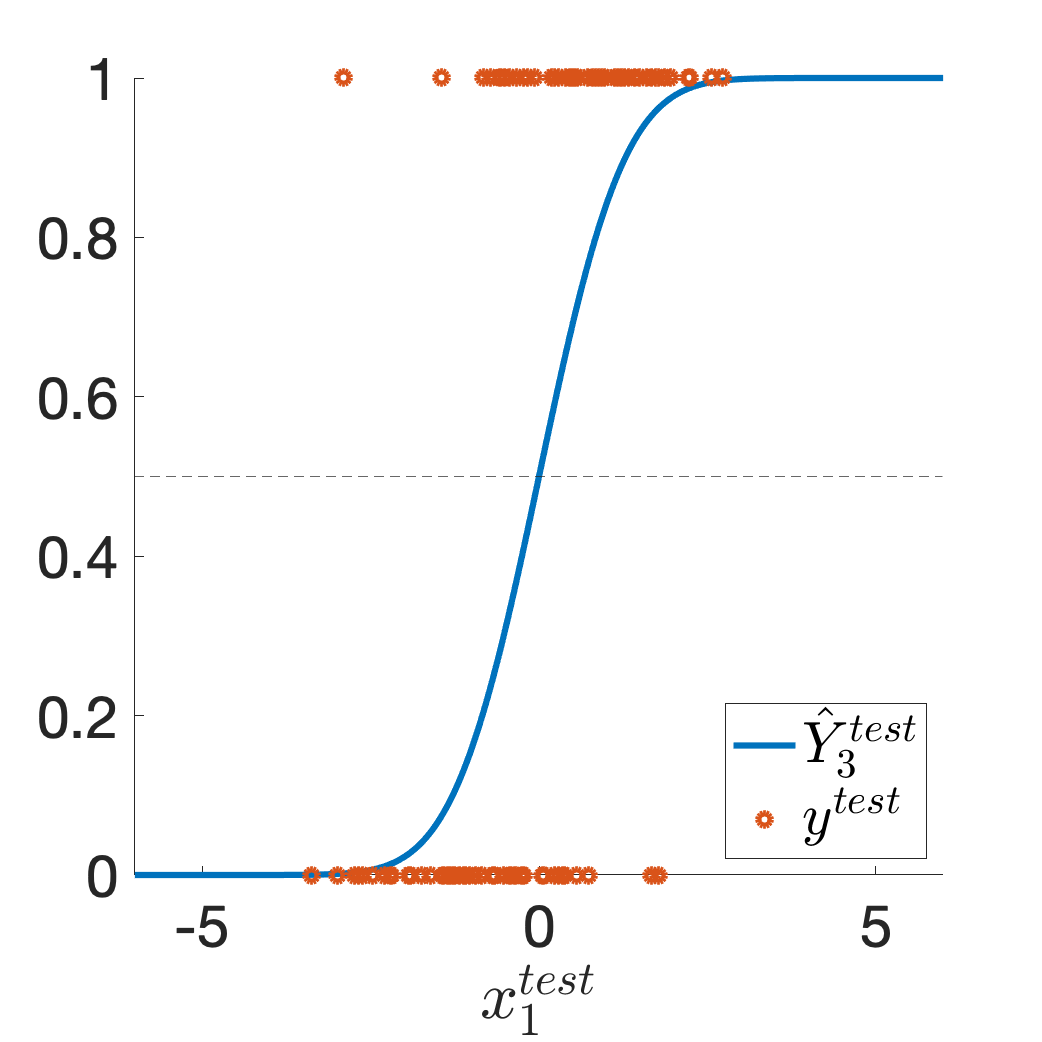}
    \label{fig:EYgZ1X3}
  \end{subfigure}
  \begin{subfigure}[b]{0.45\columnwidth}
\vspace{-2em}
\includegraphics[width=\linewidth]{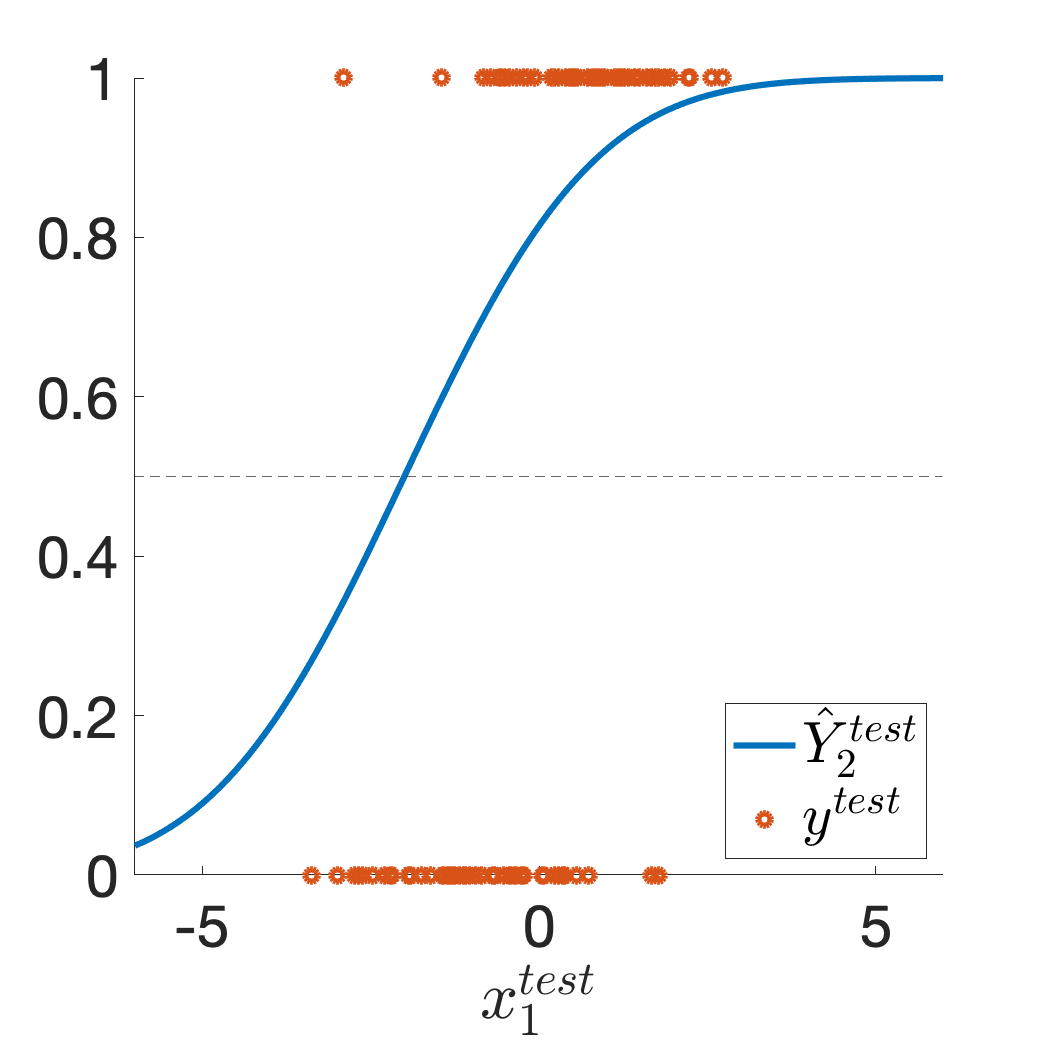}
    \label{fig:EYgZ2X3}
  \end{subfigure}
  \vspace{-1.5em}
  \caption{Comparisons of $\hat{Y}^{\text{test}}_3$ (left) and $\hat{Y}^{\text{test}}_2$ (right), where $\beta_1^e = 2$, $\mu^e_2 = 1$, $\beta_2^{\text{test}} = 0$, and $\mu^{\text{test}}_2 = -1$.}\label{fig:example}
  \vspace{-2em}
\end{figure}
\end{center}
\section{The Binary Invariant Matching Property}
A deterministic relationship such as the one in~\eqref{equ::examp_matching} has been previously referred to as \emph{matching}~\cite{du2023learning}, and can be generalized to the formulation outlined in Section~\ref{sec::formulation}.
\begin{definition}
    For $k\in\{1,\dots,m\}$, $S \subseteq \{1,\dots,m\}\backslash k$, and $h(X_S,Y) := \E_{\Pc_e}[X_k|X_S,Y]$, the pair $(k,S)$ satisfies the binary invariant matching property (bIMP)\footnote{There are \emph{degenerate} cases when $h(X_S,0) = h(X_S,1)$, for which the tower property implies $\E_{\Pc_e}[X_k|X_S] = \E_{\Pc_e}[h(X_S,Y)|X_S] = h(X_S,0)$, and the ratio in~\eqref{equ::EjgXs} reduces to $0$ divided by $0$.} if,
    \begin{equation}\label{equ::EjgXs}
    \E_{\Pc_e}[Y|X_S] = \frac{\E_{\Pc_e}[X_k|X_S] - h(X_S,0)}{h(X_S,1) - h(X_S,0)}, 
    \end{equation}
    holds for all $e\in\Ec_{\text{obs}}$, where $h(X_S,Y)$ does not depend on $e$.
\end{definition}


As seen in the example, there are a variety of choices for $k$ and $S$, not all of which lead to invariant representations. We now detail the sufficient conditions for which a pair $(k,S)$ satisfies the bIMP (see Appendix for the proof).

\begin{proposition} \label{prop::suff}
    Let $ k \in \{1,\ldots,m\}$ and $S = R \cup Q$ where $R,Q \subseteq \{1,\ldots,m\} \setminus k$ and $R \cap Q = \varnothing$. The pair $(k,S)$ satisfies the bIMP if, for every $e\in\Ec_{\text{obs}}$, 
    \begin{enumerate}
        \item $X^e_k = g(X^e_R,Y^e) + \epsilon^e$ as in~\eqref{equ::Xi_def}\ ,
        \item  $X^e_Q \independent X^e_k \ |\  (X^e_R,\ Y^e)$\ .
    \end{enumerate}
\end{proposition}

What remains is to show that the bIMP can be used to satisfy the invariance principle in~\eqref{equ::invar}, and thus, can be beneficial in predicting on unknown environments, as shown below.

\begin{theorem} \label{the:1} 
Let $ k \in \{1,\ldots,m\}$ and $S = R \cup Q$ where $R,Q \subseteq \{1,\ldots,m\} \setminus k$ and $R \cap Q = \varnothing$.
When $\phi_e(X) = (X_R,X_Q,\E_{\Pc_e}[X_k|X_R,X_Q])$,~\eqref{equ::invar} holds if the pair $(k,S)$ satisfies the bIMP. 
\end{theorem}
\begin{proof}

 Let $\ell^e(X_R,X_Q):=\E_{\Pc_e}[X_k|X_R,X_Q]$ and $\phi_e(X) = (X_R,X_Q,\ell^e(X_R,X_Q))$. Since $(k,S)$ satisfies the bIMP and $\ell^e(X_R,X_Q)$ is a function of $X_R$ and $X_Q$,
\begin{align*}
      \E_{\Pc_e}&[Y|\phi_e(X) = (x_Q,x_R,z)] \\
      &=\E_{\Pc_e}[Y|X_R=x_R,X_Q=x_q,\ell^e(X_R,X_Q)=z] \\
      &= \frac{z - g(x_R,0)}{g(x_R,1) - g(x_R,0)}. \numberthis\label{equ::E_final}
\end{align*}
Thus,~\eqref{equ::invar} holds as~\eqref{equ::E_final} does not vary over $e\in\Ec_{\text{obs}}$.  
\end{proof}

\begin{remark}
    It turns out that the form of bIMP has been studied in a concurrent work~\cite{tasche2023sparse} under the name of sparse joint shift (SJS), first developed in~\cite{chen2022estimating}. However, our proposed methods in Section~\ref{sec:method} as well as the causal perspective are completely different from that line of research.
\end{remark}

\subsection{A Causal Perspective}\label{sec::causal}
While the sufficient conditions in Theorem~\ref{the:1} may seem abstract, we now show that, in fact, they have a specific meaning in a causal sense. To do so, we introduce the structural causal model (SCM)~\cite{pearl2009causality}. Here, $X^e$ and $Y^e$ are part of an SCM $\Sc^e$ that varies over environment such that 
\begin{equation} \label{equ:SCM}
    \Sc^e: \begin{cases}
    Y^e := f_Y^e(X^e_{PA(Y^e)}\, ,\ \epsilon^e_Y),\\
    X_1^e := f_1^e(X^e_{PA(X_1^e)}\, ,\ \epsilon^e_1),\\
    \qquad \vdots \\
    X_m^e := f_m^e(X^e_{PA(X_m^e)}\, ,\ \epsilon^e_m).	\end{cases}
\end{equation}
where $\epsilon_1^e\dots\epsilon_m^e,\epsilon_Y^e$ are independent noise variables. To simplify notation, let $X_0^e:=Y^e$. Thus, $PA(X_i^e)\subseteq\{0,\dots,1\}$ denotes the set indexed by the direct causal parents of $X_i^e$ for all $i\in\{0,\dots,m\}$.

As in Section~\ref{sec::formulation}, $Y^e$ is binary. Additionally, at least one structural assignment (i.e., $f_i^e(\cdot)$) in $\Sc^e$ is an additive noise function that does not vary over environment. Specifically, for some $i\in\{0,\dots,m\}$,  let $f_i^e(X_{PA(X_i^e)}^e,\epsilon^e_i) = g(X_{PA(X_i^e)}^e) + \epsilon^e_i$, where $\epsilon^e_i$ has zero mean. An intervention on a variable from $\{X_1^e,\dots,X_m^e,Y^e\}$ occurs if the structural assignment changes for some $e\in\Ec$. Relating the SCM to the formation in Section~\ref{sec::formulation} gives insight into the types of interventions that may occur.  While many methods~\cite{peters2016causal,pfister2021stabilizing,du2023learning} make various assumptions on the types of interventions (e.g., shifts in the mean or variance), the setting in~\eqref{equ:SCM} allows for very general interventions, including general interventions on $Y^e$, which many other approaches do not allow.

Given  $\Sc^e$ for all $e\in\Ec_{\text{obs}}$, we can express the conditions of Proposition~\ref{prop::suff} in the language of SCMs, detailed below.

\begin{corollary}\label{cor:causal1}
Let $ k \in \{1,\ldots,m\}$ and $S = R \cup Q$ where $R,Q \subseteq \{1,\ldots,m\} \setminus k$ and $R \cap Q = \varnothing$. For the SCM $\Sc^e$, the pair $(k,S)$ satisfies the bIMP for all $e\in\Ec_{\text{obs}}$ if the following cases hold. 
\begin{enumerate}
    \setlength\itemsep{0em}
    \item $X_k^e = g(X^e_{PA(X_k^e)}) + \epsilon^e_k$\ ,
    \item $X_R^e$ and $Y^e$ constitute the parents of $X_k^e$\ ,
    \item The variables in $X_Q^e$ can be any non-descendants of $X_k^e$.
\end{enumerate}
\end{corollary}
The first condition in Proposition~\ref{prop::suff} is analogous to the first and second condition above as $PA(X_k^e) = (X_S,Y)$. Additionally, in an SCM, any variable conditioned on its parents is independent of any non-descendant. As such, the set $X_Q^e$ can be any non-descendant of $X_k^e$, bridging the final conditions in Proposition~\ref{prop::suff} and Corollary~\ref{cor:causal1}. 

In many cases, the set $Q$ can be quite inclusive despite what may seem like a strong independence condition in Proposition~\ref{prop::suff}. In Corollary~\ref{cor:causal1}, we learn that, in a causal sense, $X_Q^{e}$ can be any non-descendant of $X_k^e$. For example, if half of the predictors in an SCM are ancestors of $Y^e$, while the other half are descendants, then the set $Q$ indexes at least half of all predictors (and potentially many more).


\section{Proposed Method}
\label{sec:method}
For each $e\in\Ec_{\text{train}}$, we have $n_e$ samples, represented as a matrix $\bm{X}^e\in\mathbb{R}^{n_e\times m}$, and a vector $\bm{Y}^e\in \{0,1\}^{n_e}$ (see~\cite{goddard2023error} for a discussion on the impact of different environments). Additionally, we have $n_{\text{test}}$ samples in the test environment, and we denote $\bm{X}^{\text{test}}\in\mathbb{R}^{n_{\text{test}}\times m}$ and $\bm{Y}^{\text{test}}$ as the predictor matrix and target vector for the environment $e^{\text{test}}$, respectively. We denote $\bm{X}$ as the pooled predictor matrix over all $e\in\Ec_{\text{train}}$, and $\bm{X}_{Y=y}$ as the matrix comprising the rows of $\bm{X}$ in which $Y=y$, for $y\in\{0,1\}$. Let $\bm{X}^{-e}$ be the matrix of samples indexed only by those samples not in $e\in\Ec_{\text{train}}$. 

We now leverage insights gained from Theorem~\ref{the:1} and the bIMP to develop a practical method for estimation in unknown environments. At test time, we do not have access to $Y^{\text{test}}$. As such, one cannot say with definitive assurance that~\eqref{equ::invar} holds for all $e\in\Ec_{\text{obs}}$. Thus, the best that can be done in such settings is to identify a $\phi_e$ such that~\eqref{equ::invar} holds for all $e\in\Ec_{\text{train}}$, implying that $\Ec_{\text{train}}$ must have at least two environments.

Thus, our goal in a practical setting is to identify $(k,S)$ pairs that may satisfy the bIMP overall $e\in\Ec_{\text{train}}$. Simply put, we test whether $\E_{\Pc_e}[X_k|X_S,Y]$ is invariant. To do so, we consider a special form of the model in~\eqref{equ::Xi_def} where $X_k^e = g(X^e_S,Y^e) + \epsilon^e$ with $\epsilon^e\sim \Norm(0, (\sigma^e)^2)$ is assigned a different nonlinear additive noise function for each value of $Y^e$. Specifically, 
\begin{equation}\label{equ::model_Z}
    g(X_S^e,Y^e) =\begin{cases}
			g_1(X_S^e), & \text{if $Y^e = 1$},\\
            g_0(X_S^e), & \text{if $Y^e = 0$}.
		 \end{cases}
\end{equation}
As $X_k^e$ can be split into two models, one for each value of $Y^e$, we can perform an invariance test on each model. If both are found to be invariant, we can consider $\E_{\Pc_e}[X_k|X_S,Y]$ as a whole to be invariant. Invariance tests on additive noise models have been widely studied: Various tests have been proposed for linear~\cite{peters2016causal} and nonlinear~\cite{heinze2018invariant} models. We adopt one such approximate test from~\cite{heinze2018invariant} known as the \emph{residual distribution test} for our setting, as further detailed in Algorithm~\ref{alg:var_test}. 
\vspace{-.5em}
\begin{algorithm}
\caption{Binary Invariant Residual Distribution Test}\label{alg:var_test}
\textbf{Input:} $\bm{Y}^e$ and $\bm{X}^e$, for each $e\in\Ec_{\text{train}}$, significance level $\alpha$, and the pair $(k,S)$
\textbf{Output:} \emph{accepted} or \emph{rejected}
\begin{algorithmic}
 \State Regress $\bm{X}_{k,Y=i}$ on $\bm{X}_{S,Y=i}$ to get $\hat{g}_i$, for $i\in\{0,1\}$
 \For{each $e\in\Ec_{\text{train}}$ and $i\in\{0,1\}$}
\State $\bm{R}^e_i = \bm{X}^e_{k,Y=i} - \hat{g}_i(\bm{X}^e_{S,Y=i})$
\State $\bm{R}^{-e}_i = \bm{X}^{-e}_{k,Y=i} - \hat{g}_i(\bm{X}^{-e}_{S,Y=i})$
\State $\text{pval}_i^e = t\text{-test}(\bm{R}^{e}_i,\bm{R}^{-e}_i)$
\EndFor
\State Combine p-values in $\text{pval}_1^e$ and $\text{pval}_0^e$ via Bonferroni correction 
\If{$\min_{e\in\Ec_{\text{train}}}\text{pval}_1^e > \alpha$ \textbf{ and } $\min_{e\in\Ec_{\text{train}}}\text{pval}_2^e > \alpha$} 
\State return \emph{accepted}
\Else $ $ return \emph{rejected}
\EndIf
\end{algorithmic}
\end{algorithm}
\vspace{-1em}

We use Algorithm~\ref{alg:var_test} as an approximate test for whether $\E_{\Pc_e}[X_k|X_S,Y]$ is invariant over environments. We now employ this test to develop a practical method for estimating $Y^{\text{test}}$ which we refer to as \textsf{bIMP}. We adopt a similar approach to that of~\cite{pfister2021stabilizing} and~\cite{du2023learning} in which we test the invariance of $\E_{\Pc_e}[X_k|X_S,Y]$ for all possible pairs $(k,S)$. We then train models using the $X_k^e$ and $X^e_S$ which are \emph{accepted} according to Algorithm~\ref{alg:var_test}. Our bIMP models are a combination of two separate models trained to estimate both $\E_{\Pc_e}[X_k|X_S,Y]$ and $\E_{\Pc_e}[X_k|X_S]$. Given both of these estimates, we compute an estimate of $Y^{\text{test}}$ using~\eqref{equ::EjgXs}. As it is likely that more than one pair is accepted, the final estimate of $Y^{\text{test}}$ is the average estimate over all accepted pairs. 

While we can guarantee invariance via the bIMP, there is no guarantee that the estimation will predict well on $e^{\text{test}}$. As such, in addition to filtering pairs based on invariance, \textsf{bIMP} also filters based on a prediction score. Invariant pairs $\Tc_{\text{inv}}$ computed using~\eqref{equ::EjgXs} are filtered using the mean squared prediction error. The threshold by which the pairs are filtered is identical to the procedure proposed in~\cite{pfister2021stabilizing}. 

\begin{algorithm}
\caption{\textsf{bIMP}}\label{alg:method}
 \textbf{Input:} $\bm{Y}^e$, for each $e\in\Ec_{\text{train}}$, and $\bm{X}^e$, for each $e\in\Ec_{\text{obs}}$

\textbf{Output:} Estimate $\bm{\hat{Y}}^{\text{test}}$
\begin{algorithmic}
\State   Identify the set of all invariant pairs $\Tc_{\text{inv}}$ using Algorithm~\ref{alg:var_test}
\State  Filter pairs from $\Tc_{\text{inv}}$ based on prediction score 

\For{each $(k,S)$ in $\Tc_{\text{inv}}$}
\State Estimate $\E_{\Pc_e}[X_k|X_S,Y]$ by regressing $\bm{X}_k$ on $(\bm{X}_S,\bm{Y})$
\State Estimate $\E_{\Pc_{\text{test}}}[X_k|X_S]$ by regressing $\bm{X}_k^{\text{test}}$ on $\bm{X}^{\text{test}}_S$ 
\State Using~\eqref{equ::EjgXs}, compute $\bm{\hat{Y}}_{k,S}^{\text{test}}$ for the pair $(k,S)$
\EndFor
\State $\bm{\hat{Y}}^{\text{test}} = \frac{1}{|\Tc_{\text{inv}}|}\sum_{(k,S)\in\Tc_{\text{inv}}} \bm{\hat{Y}}_{k,S}^{\text{test}}$
\end{algorithmic}
\end{algorithm}
The \textsf{bIMP} method proposed gives freedom to the user to select the underlying models with which to estimate $\E_{\Pc_e}[X_k|X_S]$ and $\E_{\Pc_e}[X_k|X_S,Y]$. In the case of $\E_{\Pc_e}[X_k|X_S]$, we have complete freedom to select whichever model suits the data, be it linear or nonlinear.  For $\E_{\Pc_e}[X_k|X_S,Y]$, we are restricted by the additive noise of~\eqref{equ::Xi_def}. In addition, we have chosen to model $X_k$ using two sub-models, one for each value of $Y$ as in~\eqref{equ::model_Z}. This, however, is not the only option and depends on the invariance test used. When estimating each model, ordinary least squares (OLS) could be used for linear models, and a generalized additive model (GAM) or Gaussian process regression could be used for nonlinear models. In practice, we found estimating each model using OLS to be the most efficient, as fitting two nonlinear models for all possible $(k,S)$ pairs can be computationally expensive. 

\begin{remark}
There are several challenges with this approach that we leave for future work. We observe that nonlinear implementations of the invariance test (Algorithm~\ref{alg:var_test}) may lead to erroneously accepted invariant pairs. In addition to this, the complexity of training a nonlinear model for all possible $(k,S)$ pairs can be high. Finally, the effects of model misspecification can be challenging to analyze.
\end{remark}
\section{Experiments}
We provide one synthetic and two real datasets to test the effectiveness of \textsf{bIMP} and compare with the following two baselines: (1) a binary adaptation of Method~II from~\cite{peters2016causal} (\textsf{ICP}), and (2) logistic regression (\textsf{LR}). While we do not expect \textsf{LR} to perform well on unknown environments, it serves as a natural baseline. While \textsf{ICP} can handle the binary response setting via logistic regression, \textsf{SR} is specific to regression settings and thus not reported. In all experiments, we set $\alpha = 0.1$. 

\begin{center}
\begin{figure}[h!]
\centering
\vspace{-1.5em}
  \includegraphics[width=0.85\linewidth]{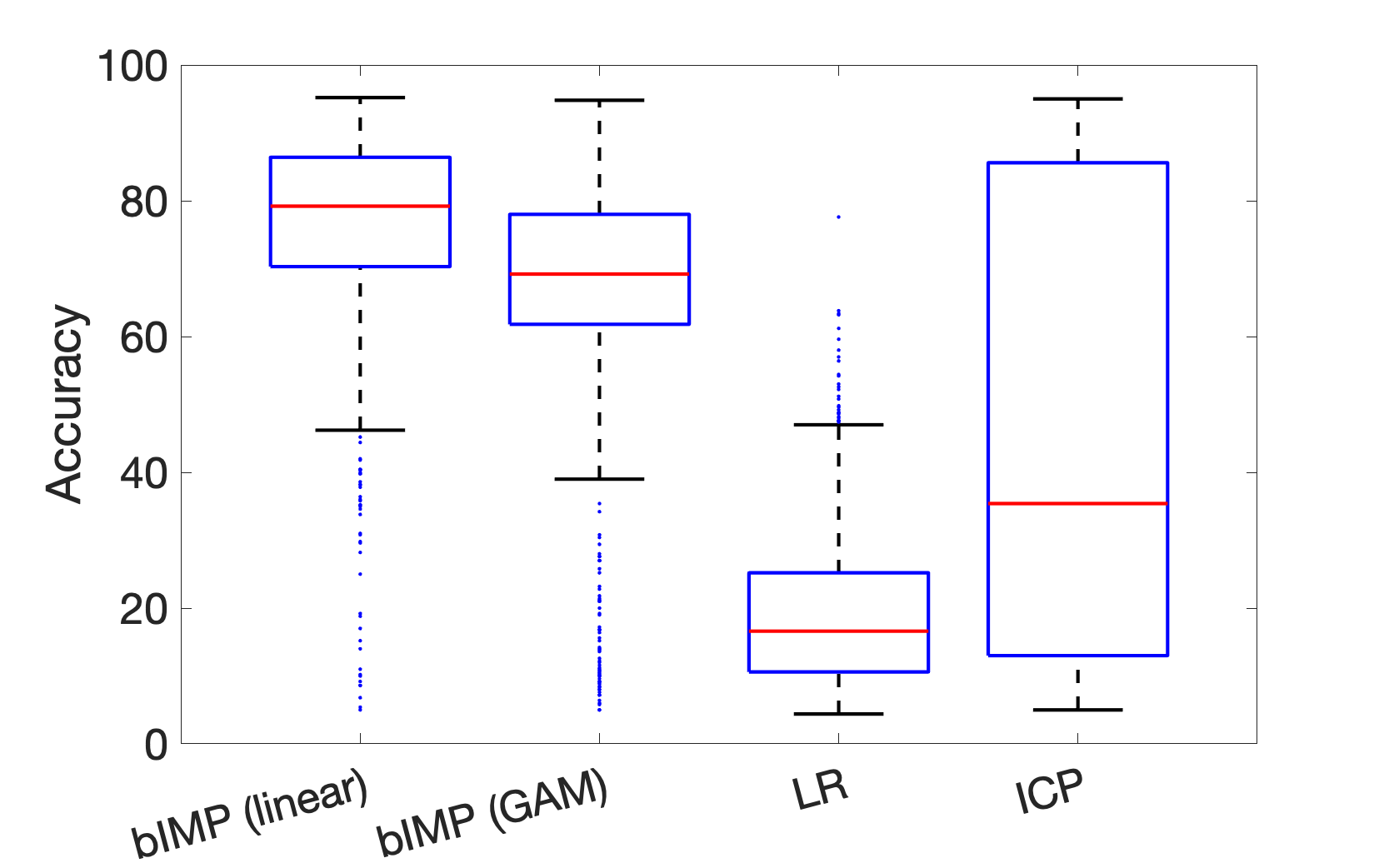}
  \caption{Simulation accuracy over $1000$ simulated datsets.}\label{fig:sim}
  \vspace{-2em}
\end{figure}
\end{center}

As there is some degree of freedom in selecting how the sub-models in \textsf{bIMP} are trained, we explore two variants of \textsf{bIMP}: \textsf{bIMP} (linear) and \textsf{bIMP} (GAM). For both variants, we follow the invariance test in Algorithm~\ref{alg:var_test} and estimate $g_1$ and $g_0$ using OLS. We estimate $\E_{\Pc_e}[X_k|X_S]$ using OLS for \textsf{bIMP} (linear), and a GAM for \textsf{bIMP} (GAM).


\noindent\underline{\bf Synthetic data.} The simulated dataset is generated as follows. We generate data from three environments, $e^1,e^2\in\Ec_{\text{train}}$, and $e^{\text{test}}$. The number of predictors $m$ is randomly selected from $\{3,\dots,7\}$. For each $i \in \{2,\dots,m\}$ and $e\in\Ec_{\text{train}}$, $X_i^e \sim \mathcal{N}(\mu^e_i,1)$, and $\mu^e_i$ is randomly selected on the interval $[-2,0]$ for $e=e^1$,  $[0,2]$ for $e=e^2$, and $[0,3]$ for $e=e^{\text{test}}$. Then, where $S_1 = \{2,\dots,m\}$, $Y^e|X_{S_1}^e$ follows a logistic model such that $\Pc_e(Y = 1 | X_{S_1}) = 1/(1+e^{-X_{S_1}\beta^e})$ for $e\in\Ec_{\text{train}}$. For $e^{\text{test}}$, $Y^{\text{test}}|X_{S_1}^{\text{test}}$ follows a probit model such that $Y^{\text{test}}=1$, if $X_{S_1}^{\text{test}}\beta^{\text{test}} + \epsilon < 0$, where $\epsilon\sim\mathcal{N}(0,1)$. For all $e\in\Ec_{\text{obs}}$, randomly select $\beta^e$ as $\beta^e\sim \U[0,1]$. The coefficients are then scaled such that they sum to one. For all $e\in\Ec_{\text{obs}}$, the variable $X_1^e$ is then simulated similarly to $X_k^e$ in~\eqref{equ::model_Z}. Specifically, $g_1(X_{S_1}^e) = X_{S_1}^e\eta_1$ and $g_0(X_{S_1}^e) = X_{S_1}^e\eta_0$. The noise term associated with $X_1^e$ is a standard normal. The coefficients $\eta_{1,i}\sim \U[0,1]$ and $\eta_{0,i}\sim \U[0,1]$ do not vary over the environment. The number of samples per environment is fixed to $1000$.

Simulation results on both accuracy and mean squared error (MSE) indicate that \textsf{bIMP} can generalize to the test environment while \textsf{LR} and \textsf{ICP} are not (Fig~\ref{fig:sim}). In addition, while we expect \textsf{LR} to behave poorly, \textsf{ICP} also performs poorly as all parents of $Y$ are intervened in every simulation. We note that while \textsf{bIMP} (linear) slightly outperforms \textsf{bIMP} (GAM) in this setting, it is not likely always the case. Particularly, we see that when $\eta_{0,i}$ and $\eta_{1,i}$ are similar, $\E_{\Pc_e}[X_k|X_S,Y]$ is approximately linear. This can be seen by solving for $\E_{\Pc_e}[X_k|X_S,Y]$ in~\eqref{equ::EjgXs}. In this case, $h(X_S,0)$ (a linear function) will dominate and OLS will provide a good estimate. The long version of this work will carry out extensive experiments to reveal the benefits of using GAMs. 

\vspace{-1em}
\begin{center} 
\begin{table}[h!]
\begin{center}
{\small
\begin{tabular}{ |c||c|c|c|c|  }
\hline 
  & \textsf{bIMP} (linear) & \textsf{bIMP} (GAM)  & \textsf{LR}\\ 
 \hline
 \hline
 Environment & \multicolumn{3}{|c|}{Accuracy} \\
 \hline
 born in US & $85.0$ & $84.9$ & $78.2$ \\
 overtime   & $68.4$ & $59.1$ & $77.0$ \\
 caucasian  & $85.0$ & $85.2$ & $78.1$ \\
 \hline
\end{tabular}
} 
\caption{\emph{census}: performance and training environments.}\label{tbl::census}
\end{center}
\end{table}
\end{center}
\vspace{-1.5em}

\noindent\underline{\bf Two real-world data.} We also include experiments on two real datasets: \emph{census}~\cite{misc_adult_2} and \emph{mushroom}~\cite{misc_mushroom_73}. The census dataset is data gathered from the $1994$ US census and contains $14$ societal and demographic variables such as age, education, marital status, and working class. The target variable used is whether or not an individual's income exceeded $50$k/yr. The data is first split into test and training data by whether or not a person graduated from a college. Thus, we train only on those who did not graduate college with the aim of extending our trained model to those who did. We further split the training data and run the methods on each set of training environments. The variables used to split the training data into environments are ``was the person born in the US",  ``do they regularly work more than $40$hr/week", and ``does the person identify as Caucasian". The experiment shows that \textsf{bIMP} outperforms \textsf{LR} and \textsf{ICP} in all environments aside from the \emph{overtime} environment (Table~\ref{tbl::census}). The \textsf{ICP} method returns no invariant predictors for any environment, thus no predictions can be made and no accuracy is reported; this is also the case for the \emph{mushroom} data below.

\vspace{-1em}
\begin{center} 
\begin{table}[h!]
\begin{center}
{\small
\begin{tabular}{ |c||c|c|c|  }
\hline 
  & \textsf{bIMP} (linear) & \textsf{bIMP} (GAM) & \textsf{LR} \\ 
 \hline
 \hline
 Environment & \multicolumn{3}{|c|}{Accuracy} \\
 \hline
 meadows & $76.0$  & $87.5$ & $46.2$ \\
 paths   & $88.1$  & $90.9$ & $11.8$ \\
 \hline
\end{tabular}
} 
\caption{\emph{mushroom}: performance and training environments.}\label{tbl::mush}
\vspace{-2em}
\end{center}
\end{table}
\end{center}

The \emph{mushroom} dataset contains $16$ features related to naturally growing mushrooms' size, shape, and color and showcases how the proposed approach can handle discrete and categorical data. We aim to predict whether or not a mushroom is edible based on these factors. The environments on which we predict are the habitats in which the mushrooms grow. Specifically, we train on mushrooms that grow in grass or urban habitats and test on mushrooms that grow in meadows or paths. Results in Table~\ref{tbl::mush} indicate that \textsf{bIMP} outperforms \textsf{ICP} and \textsf{LR} for both the linear and GAM variants, while the GAM variant performed the best. 

\section{Acknowledgements}
 We thank the anonymous reviewers for their helpful comments that improved the quality of this work.

\appendix
\begin{proof}[Proof of Proposition~\ref{prop::suff}]
    First, we show that~\eqref{equ::EjgXs} holds for any $e\in\Ec_{\text{obs}}$.  Without loss of generality, let $X^e_i$ be continuous for all $i\in\{1,\dots,m\}$. The pdf of $X_k^e|X_S^e$ for any $e\in\Ec_{\text{obs}}$ is
\begin{align*}
    f&_{X_k^e|X_S^e}(x_k|x) \\
    &= f_{X_k^e|X_S^e,Y^e}(x_k|x,1)\cdot p_{Y^e|X_S^e}(1|x) \\
    &\qquad +  f_{X_k^e|X_S^e,Y^e}(x_k|x,0)\cdot\left[1-p_{Y^e|X_S^e}(1|x)\right] \\
    & = p_{Y^e|X_S^e}(1|x)\left[f_{X_k^e|X_S^e,Y^e}(x_k|x,1) - f_{X_k^e|X_S^e,Y^e}(x_k|x,0)\right] \\
    &\qquad + f_{X_k^e|X_S^e,Y^e}(x_k|x,0).\numberthis\label{equ::fxi|xS}
\end{align*}
Then using~\eqref{equ::fxi|xS}, we can write $\E_{\Pc_e}[X_k|X_S=x]$ as
\begin{align*}
     &\int_{-\infty}^{\infty} x_k\cdot f_{X_k^e|X_S^e}(x_k|x) \,dx_k \\
    &= \E_{\Pc_e}[Y|X_S=x]\cdot\E_{\Pc_e}[X_k|X_S=x,Y=1] \\
    &\qquad - \E_{\Pc_e}[Y|X_S=x]\cdot\E_{\Pc_e}[X_k|X_S=x,Y=0] \\
    &\qquad + \E_{\Pc_e}[X_k|X_S=x,Y=0] . \numberthis\label{equ::EXjgXs}
\end{align*}
Thus, $\E_{\Pc_e}[Y|X_S]$ can be written as
\begin{equation}\label{equ::EjgXs_ther}
     \frac{\E_{\Pc_e}[X_k|X_S] - \E_{\Pc_e}[X_k|X_S,Y=0]}{\E_{\Pc_e}[X_k|X_S,Y=1] - \E_{\Pc_e}[X_k|X_S,Y=0]}.
\end{equation}

We  now show (I) $\E_{\Pc_e}[X_k|X_S=x,Y=y]$ does not depend on $e$ and (II) the denominator of~\eqref{equ::EjgXs_ther} is non-zero. Since $X_S^e = (X_R^e,X_Q^e)$, 
\begin{align*}
\E_{\Pc_e}[X_k&|X_S,Y]  = \E_{\Pc_e}[X_k|X_R,X_Q,Y] \overset{(a)}{=} \E_{\Pc_e}[X_k|X_R,Y] \\
   & \overset{(b)}{=} \E_{\Pc_e}[g(X_R,Y) + \epsilon |X_R,Y] = g(X_R^e,Y^e), \numberthis
\end{align*}
where $(a)$ follows since $X_Q^e \independent X_k^e | X_R^e,Y^e$, $(b)$ follows from the assumption $X^e_k = g(X_R^e,Y^e) + \epsilon^e$, and $(c)$ follows since $\epsilon$ has zero mean. Thus, the $\E_{\Pc_e}[X_k|X_S=(x_Q,x_R),Y=y]$ does not depend on $e$ as $\E_{\Pc_e}[X_k|X_S=(x_Q,x_R),Y=y] = g(x_R,y)$. As the output of the function $g$ is not constant with regards to any of its inputs as in~\eqref{equ::Xi_def}, the denominator of~\eqref{equ::EjgXs_ther} is non-zero.
\end{proof}

\balance



\balance

\newpage
\bibliographystyle{IEEEtran}
\bibliography{sample}

\end{document}